\documentclass[a4paper,onecolumn,superscriptaddress,11pt,accepted=2018-01-16]{quantumarticle}

\usepackage[numbers,sort&compress]{natbib}

\usepackage{amsmath} 
\usepackage{bbm} 
\usepackage{amsfonts} 
\usepackage{xparse} 
\usepackage[margin=1.5in]{geometry} 
\usepackage{enumerate}
\usepackage{mathtools}
\usepackage{xcolor}

\usepackage{amsthm}
\theoremstyle{definition}

\theoremstyle{plain}
\newtheorem{theorem}{Theorem}
\newtheorem{proposition}[theorem]{Proposition}
\newtheorem{lemma}[theorem]{Lemma}
\newtheorem{corollary}[theorem]{Corollary}

\theoremstyle{definition}

\newtheorem{remark}[theorem]{Remark}

\usepackage{caption}
\usepackage{subcaption}
\usepackage{setspace}

\usepackage{multirow}
\usepackage{tikz}
\usepackage{caption}
\usepackage{subcaption}
\newsavebox{\tempbox}
\usepackage{authblk}

\usepackage{hyperref} 

\def\complex{\mathbb{C}}

\def\ran{\textnormal{ran}}

\def\Tr{\textnormal{Tr}}

\newcommand\I{\mathbbm{1}}
\newcommand{\id}[1]{\textnormal{id}_#1}
\let\originalleft\left
\let\originalright\right

\renewcommand{\left}{\mathopen{}\mathclose\bgroup\originalleft}
\renewcommand{\right}{\aftergroup\egroup\originalright}

\newcommand{\norm}[1]{\| #1 \|}
\newcommand{\bignorm}[1]{\big\| #1 \big\|}
\newcommand{\Bignorm}[1]{\Big\| #1 \Big\|}
\newcommand{\biggnorm}[1]{\bigg\| #1 \bigg\|}

\newcommand{\vertiii}[1]{{\vert\kern-0.25ex\vert\kern-0.25ex\vert #1 
    \vert\kern-0.25ex\vert\kern-0.25ex\vert}}
\newcommand{\bigvertiii}[1]{{\big\vert\kern-0.25ex\big\vert\kern-0.25ex\big\vert #1 
    \big\vert\kern-0.25ex\big\vert\kern-0.25ex\big\vert}}
\newcommand{\Bigvertiii}[1]{{\Big\vert\kern-0.25ex\Big\vert\kern-0.25ex\Big\vert #1 
    \Big\vert\kern-0.25ex\Big\vert\kern-0.25ex\Big\vert}}
\newcommand{\biggvertiii}[1]{{\bigg\vert\kern-0.25ex\bigg\vert\kern-0.25ex\bigg\vert #1 
    \bigg\vert\kern-0.25ex\bigg\vert\kern-0.25ex\bigg\vert}}
\newcommand{\Biggvertiii}[1]{{\Bigg\vert\kern-0.25ex\Bigg\vert\kern-0.25ex\Bigg\vert #1 
    \Bigg\vert\kern-0.25ex\Bigg\vert\kern-0.25ex\Bigg\vert}}

\newcommand{\ip}[2]{\langle #1, #2 \rangle}

\NewDocumentCommand\linear{mo}{
  \IfNoValueTF{#2}
	      {\text{L}(\mathcal{#1})}
	      {\text{L}(\mathcal{#1}, \mathcal{#2})}
}

\renewcommand{\t}{{\scriptscriptstyle\mathsf{T}}}

\title{Characterization of linear maps on $M_n$ whose multiplicity maps have maximal norm, with an application in quantum information}
\author{Daniel Puzzuoli}
\orcid{0000-0002-3815-2072}
\affiliation{Department of Applied Mathematics and Institute for Quantum Computing\\University of Waterloo, Waterloo, Ontario, Canada}

\begin{document}
\maketitle

\begin{abstract}
Given a linear map $\Phi : M_n \rightarrow M_m$, its multiplicity maps are defined as the family of linear maps $\Phi \otimes \id{k} : M_n \otimes M_k \rightarrow M_m \otimes M_k$, where $\id{k}$ denotes the identity on $M_k$. Let $\norm{\cdot}_1$ denote the trace-norm on matrices, as well as the induced trace-norm on linear maps of matrices, i.e. $\norm{\Phi}_1 = \max\{\norm{\Phi(X)}_1 : X \in M_n, \norm{X}_1 = 1\}$. A fact of fundamental importance in both operator algebras and quantum information is that $\norm{\Phi \otimes \id{k}}_1$ can grow with $k$. In general, the rate of growth is bounded by $\norm{\Phi \otimes \id{k}}_1 \leq k \norm{\Phi}_1$, and matrix transposition is the canonical example of a map achieving this bound. We prove that, up to an equivalence, the transpose is the unique map achieving this bound. The equivalence is given in terms of complete trace-norm isometries, and the proof relies on a particular characterization of complete trace-norm isometries regarding preservation of certain multiplication relations.

We use this result to characterize the set of single-shot quantum channel discrimination games satisfying a norm relation that, operationally, implies that the game can be won with certainty using entanglement, but is hard to win without entanglement. Specifically, we show that the well-known example of such a game, involving the Werner-Holevo channels, is essentially the unique game satisfying this norm relation. This constitutes a step towards a characterization of single-shot quantum channel discrimination games with maximal gap between optimal performance of entangled and unentangled strategies.
\end{abstract}

\section{Introduction}

For a linear map $\Phi : M_n \rightarrow M_m$, it is a well-known phenomenon that the norm of the multiplicity maps $\Phi \otimes \id{k} : M_n \otimes M_k \rightarrow M_m \otimes M_k$ can grow with $k$. This phenomenon has been extensively studied within the theory of C$^*$-algebras, leading to the topic of completely bounded maps \cite{smith_completely_1983,tomiyama_recent_1983,paulsen_completely_2003}. Within the field of quantum information, this phenomenon is connected to the study of entanglement. For a density matrix $\rho \in M_n \otimes M_k$, if 
\begin{equation}
	\norm{(\Phi \otimes \id{k})(\rho)}_1 > \norm{\Phi}_1,
\end{equation} 
then $\rho$ is entangled, and a well-known result in quantum information is that the existence of a positive linear map $\Phi$ for which the above holds is also necessary for $\rho$ to be entangled \cite{horodecki_separability_1996}. Indeed, one of the simplest and most well-known entanglement measures is the \emph{negativity} \cite{vidal_computable_2002}, which, up to additive and multiplicative scalars is defined as $\bignorm{(T_n \otimes \id{k})(\rho)}_1$, where $T_n$ is the transpose on $M_n$. Hence, the growth of the norms of multiplicity maps is of fundamental mathematical interest, and is deeply connected to the study and quantification of entanglement in quantum information.\footnote{Note that, in quantum information the trace-norm is typically used, whereas in operator algebras the operator norm is typically used. Due to the duality of these norms, it is possible to translate facts about one norm into facts about the other. For example, in this paper, we  reference the C$^*$-algebra literature for facts about the trace-norm, even though the trace-norm does not explicitly appear in the references. For readers unfamiliar with this duality, we describe the relevant facts in Appendix \ref{appendix:duality}.}

While the norm of $\Phi \otimes \id{k}$ may grow with $k$, the growth rate is limited. For any linear map $\Phi : M_n \rightarrow M_m$ it generically holds that
\begin{equation}
	\norm{\Phi \otimes \id{k}}_1 \leq k \norm{\Phi}_1. \label{equation:general_inequality}
\end{equation}
In this paper we are concerned with linear maps on matrices and will use the trace-norm, but we note that this fact is known much more generally (in terms of the C$^*$-norm) for maps on unital C$^*$-algebras (see \cite[Exercise 3.10]{paulsen_completely_2003}). Provided $k \leq n$, the canonical example of a map achieving equality in Equation \eqref{equation:general_inequality} is the matrix transpose $T_n$ \cite{tomiyama_transpose_1983}.

Our main result is a characterization of the maps achieving equality in Equation \eqref{equation:general_inequality}. We prove that, up to an equivalence, the transpose is in fact the \emph{only} map achieving equality in Equation \eqref{equation:general_inequality}. More specifically, a linear map $\Phi : M_n \rightarrow M_m$ satisfies Equation \eqref{equation:general_inequality} with equality if and only if there exists an isometric embedding of $M_k$ into $M_n$ on which $\Phi$ acts as the transpose followed by a complete trace-norm isometry. 

The proof relies on a characterization of complete trace-norm isometries particularly suited to the problem. This characterization (among others proved) relates to how complete trace-norm isometries preserve certain multiplication relations. For example, if a linear map $\Phi : M_n \rightarrow M_m$ is a complete trace-norm isometry, and if $A^*B = C^*D$ for $A,B,C,D \in M_n$, then $\Phi(A)^*\Phi(B) = \Phi(C)^*\Phi(D)$. These statements and their proofs are somewhat similar to multiplicative domain proofs for unital completely-positive maps on C$^*$-algebras \cite{choi_schwarz_1974} (see also \cite[Theorem 3.18]{paulsen_completely_2003}). We remark that the structure of complete trace-norm isometries on $M_n$, and consequently some of the other characterizations we give, may be deduced from the more general structure of (not necessarily complete) trace-norm isometries given in \cite{li_isometries_2005}. Nevertheless, we give self-contained proofs, and in some cases are able to utilize the ``complete'' assumption to prove certain implications in more generality (e.g. when the domain is a subspace $V \subset M_n$), which may be of independent interest.

We also apply the main result in the setting of single-shot quantum channel discrimination, which is the task of determining which of two known quantum channels is acting on a system given only a single use of the channel. As we will describe in detail in Section \ref{section:werner_holevo}, this task may be formulated as a game parametrized by a triple $(\lambda, \Gamma_0, \Gamma_1)$, where $\Gamma_0, \Gamma_1 : M_n \rightarrow M_m$ are quantum channels, and $\lambda \in [0,1]$ is a probability. Letting $\vertiii{\cdot}_1$ denote the completely bounded trace-norm (see Section \ref{section:background}), we characterize such triples satisfying the norm relation 
\begin{equation}
	1 = \vertiii{\lambda \Gamma_0 - (1-\lambda) \Gamma_1}_1 = n \norm{\lambda \Gamma_0 - (1-\lambda) \Gamma_1}_1. \label{equation:intro_werner_holevo}
\end{equation}
Operationally, the above norm relations imply that the game can be won with certainty using entanglement, but is hard to win without entanglement. In particular, we prove that the triple $(\lambda, \Gamma_0, \Gamma_1)$ satisfies Equation \eqref{equation:intro_werner_holevo} if and only if it is in some sense equivalent to a game involving the Werner-Holevo channels which is known to satisfy Equation \eqref{equation:intro_werner_holevo} (see \cite[Example 3.36]{watrous_quantum_2017}).

In Section \ref{section:background} we provide some background and definitions. In Section \ref{section:complete_trace} we prove various characterizations of complete trace-norm isometries. One characterization in particular is specially suited for later use, but we also include other characterizations; including one that may naturally be interpreted as a linear map being a complete trace-norm isometry if and only if its Choi matrix is maximally entangled. In Section \ref{section:max_norm_inflation} we prove the main result, which characterizes the linear maps $\Phi : M_n \rightarrow M_m$ for which $\norm{\Phi \otimes \id{k}}_1 = k \norm{\Phi}_1$, and as a Corollary characterize the maps for which $\vertiii{\Phi}_1 = n \norm{\Phi}_1$. In Section \ref{section:werner_holevo} we use the main result to prove that the Werner-Holevo channel discrimination game is in some sense the unique game (with input dimension $n$) satisfying Equation \eqref{equation:intro_werner_holevo}. Finally, we end with a discussion of open problems in Section \ref{section:discussion}.

\section{Notation and background} \label{section:background}

For an integer $1 \leq a \leq n$ we denote $e_a \in \complex^n$ to be the vector with a $1$ in the $a^{th}$ entry, and zeroes everywhere else. Similarly, for integers $1 \leq a,b \leq n$, we denote $E_{a,b} \in M_n$ to be the elementary matrix with a $1$ in the $(a,b)$-entry and zeroes in all other entries.

For a matrix $A$, the trace-norm is defined as $\norm{A}_1 = \Tr(\sqrt{A^*A})$, and for a linear map $\Phi :M_n \rightarrow M_m$, the induced trace-norm of $\Phi$ is given by
\begin{equation}
	\norm{\Phi}_1 = \max \{ \norm{\Phi(X)}_1 : X \in M_n, \norm{X}_1 \leq 1\}.
\end{equation}
The completely bounded trace norm is given by
\begin{equation}
	\vertiii{\Phi}_1 = \sup\big\{\norm{\Phi \otimes \id{k}}_1 : k \geq 1\big\} = \norm{\Phi \otimes \id{n}}_1.
\end{equation}
For a linear map $\Phi : M_n \rightarrow M_m$, we will use $J(\Phi)$ to denote its Choi matrix \cite{choi_completely_1975}, which we define as
\begin{equation}
	J(\Phi) = \sum_{a,b=1}^n \Phi(E_{a,b}) \otimes E_{a,b}.
\end{equation}
We will use $T_n$ to denote the transpose on $M_n$, and write $T_n(A)$ or $A^\t$ to denote the transpose of a matrix $A \in M_n$.

We will need a few concepts from quantum information, even for the sections not directly involving that topic. An element $\rho \in M_n$ is called a \emph{density matrix} if $\rho \geq 0$ and $\Tr(\rho) = 1$. A \emph{quantum channel} is a linear map $\Gamma : M_n \rightarrow M_m$ that is completely positive and trace preserving.

We will also use the term \emph{maximal entanglement}. A unit vector $u \in \complex^n \otimes \complex^m$ is called \emph{maximally entangled} if, for $r = \min(n,m)$, there exists orthonormal sets of unit vectors $\{x_a\}_{a=1}^r \subset \complex^n$ and $\{y_a \}_{a=1}^r \subset \complex^m$ for which
\begin{equation}
	u = \sqrt{\frac{1}{r}} \sum_{a = 1}^r x_a \otimes y_a. \label{equation:max_ent_vectors}
\end{equation}
As mentioned in the introduction, the \emph{negativity} of a density matrix $\rho \in M_n \otimes M_m$ is defined (up to multiplicative and additive scalars) as $\norm{(T_n \otimes \id{m})(\rho)}_1$ \cite{vidal_computable_2002}. This expression is meant to quantify the entanglement of the density matrix $\rho$. If $\rho$ is \emph{pure}, i.e. $\rho = uu^*$ for a unit vector $u \in \complex^n \otimes \complex^m$, then the negativity achieves its maximum value of $\min(n,m)$ if and only if $u$ is maximally entangled. Theorem 7 of \cite{puzzuoli_ancilla_2017}, which we now state, provides a characterization of matrices $X \in M_n \otimes M_m$ with $\norm{X}_1=1$ and satisfying $\norm{(T_n \otimes \id{m})(X)}_1 = n$. While it is not physically meaningful if $X$ is not a density matrix, the theorem loosely provides a notion of ``maximal entanglement'' for arbitrary elements of $M_n \otimes M_m$.

\pagebreak

\begin{theorem} \label{theorem:max_entangled_structure}
Let $X \in M_n \otimes M_m$ with $\norm{X}_1 \leq 1$. The following are equivalent.
\begin{enumerate}
	\item $\norm{(T_n \otimes \id{m})(X)}_1 = n$.
	\item $m\geq n$, and there exists a positive integer $r \leq m/n$, a density matrix $\sigma \in M_r$, and isometries $U,V : \complex^n \otimes \complex^r \rightarrow \complex^m$ for which
	\begin{equation}
		X = (\I_n \otimes U)(\tau_n \otimes \sigma)(\I_n \otimes V^*),
	\end{equation}
	where $\tau_n = \frac{1}{n} \sum_{a,b=1}^n E_{a,b} \otimes E_{a,b}  \in M_n \otimes M_n$ is the canonical maximally entangled state.
\end{enumerate}
If $X$ is a density matrix then the above equivalence holds with $V = U$. 
\end{theorem}

\begin{remark} \label{remark:pw_theorem}
We will make use of the additional special case of the above theorem when $X$ is Hermitian. In this case the second statement may be rewritten as: $m \geq n$, and there exists a positive integer $r \leq m/n$, a Hermitian matrix $H \in M_r$ with $\norm{H}_1 = 1$, and an isometry $U : \complex^n \otimes \complex^r \rightarrow \complex^m$ for which 
\begin{equation}
	X = (\I_n \otimes U)(\tau_n \otimes H)(\I_n \otimes U^*).
\end{equation} 
The only change necessary to the proof is to take a spectral decomposition of $X$, rather than a singular value decomposition. The rest of the proof follows as before.
\end{remark}

The last notion in quantum information that we will make use of is that of reversible quantum channels. A linear map $\Phi : M_n \rightarrow M_m$ is called a \emph{reversible quantum channel} if it is a quantum channel, and has a left inverse $\Psi : M_m \rightarrow M_n$ that is also a quantum channel.\footnote{This terminology is motivated by the fact that quantum channels model physical processes, and so a quantum channel having a left inverse that is also a quantum channel means that it can be physically undone, or reversed.} Due to its connection to error correction, conditions for reversibility (or recoverability) of a channel continue to be extensively studied in various settings (for example, see \cite{jencova_reversibility_2012,sutter_multivariate_2017}). 

\section{Complete trace-norm isometries} \label{section:complete_trace}

Let $V \subset M_n$ be a subspace, and $\Phi : V \rightarrow M_m$ be a linear map. We say that $\Phi$ is a \emph{$k$-trace-norm isometry} (or that it is \emph{$k$-trace-norm isometric}) if $\Phi \otimes \id{k} : V \otimes M_k \rightarrow M_m \otimes M_k$ is a trace-norm isometry, and say that it is a \emph{complete trace-norm isometry} (or that it is \emph{completely trace-norm isometric}) if it is a $k$-trace-norm isometry for all integers $k \geq 1$.

The purpose of this section is to give various characterizations of complete trace-norm isometries taking $M_n$ into $M_m$. Note that the structure of surjective operator norm isometries (and hence surjective complete operator norm isometries) between C$^*$-algebras is well-known \cite{kadison_isometries_1951}. Furthermore, in the matrix algebra case, a characterization of (not necessarily surjective) operator norm isometries mapping $M_n \rightarrow M_k$ has been given for the case $k \leq 2n-1$ \cite{cheung_isometries_2004}. However, the dual/adjoint of a trace-norm isometry need not be an operator norm isometry, and so it is not possible to import those results here. 

We give the various characterizations in the theorem below. Remarks and some background on what is already known, as well as some intermediate results, are given before its proof.

\newpage

\begin{theorem} \label{theorem:completely_isometric}
For a linear map $\Phi : M_n \rightarrow M_m$ the following are equivalent.
\begin{enumerate}
	\item $\Phi$ is a complete trace-norm isometry. \label{statement:complete}
	\item $\Phi$ is a 2-trace-norm isometry. \label{statement:2}
	\item For $A,B,C,D \in M_n$ the following implications hold:
		\begin{itemize}
			\item $A^*B = C^*D \implies \Phi(A)^* \Phi(B) = \Phi(C)^* \Phi(D)$, and 
			\item $AB^* = CD^* \implies \Phi(A)\Phi(B)^* = \Phi(C)\Phi(D)^*$,
		\end{itemize}
		and $\norm{\Phi(X)}_1 = \norm{X}_1$ for some $X \in M_n\setminus \{0\}$. \label{statement:multiplication_full}
	\item For $A,B \in M_n$ the following implications hold:
		\begin{itemize}
			\item $A^*B = 0 \implies \Phi(A)^*\Phi(B) = 0$, and
			\item $AB^* = 0 \implies \Phi(A)\Phi(B)^* = 0$,
		\end{itemize}
		and $\norm{\Phi(X)}_1 = \norm{X}_1$ for some $X \in M_n\setminus \{0\}$. \label{statement:multiplication_orthogonal}
	\item For rank-1 $A,B \in M_n$ the following implications hold:
		\begin{itemize}
			\item $A^*B = 0$ and $A^*A = B^*B \implies \Phi(A)^*\Phi(B) = 0$, and
			\item $AB^* = 0$ and $AA^* = BB^* \implies \Phi(A)\Phi(B)^* = 0$,
		\end{itemize}
		and $\norm{\Phi(X)}_1 = \norm{X}_1$ for some $X \in M_n\setminus \{0\}$. \label{statement:multiplication_orthogonal_restricted}
	\item $\norm{J(\Phi)}_1 = n$ and $\norm{J(\Phi T_n)}_1 = n^2$. \label{statement:choi}
	\item $m \geq n$, and there exists a positive integer $r \leq m/n$, a density matrix $\sigma \in M_r$, and isometries $U,V : \complex^n \otimes \complex^r \rightarrow \complex^m$ for which
	\begin{equation}
		\Phi(X) = U(X \otimes \sigma)V^*
	\end{equation}
	for all $X \in M_n$. \label{statement:structure}
	\item $\vertiii{\Phi}_1 = 1$, and $\Phi$ has a left inverse $\Psi :M_m \rightarrow M_n$ with $\vertiii{\Psi}_1 =1$. \label{statement:left_inverse}
\end{enumerate}
If, in addition, $\Phi$ is positive, then statement \ref{statement:structure} holds with $V =U$, making $\Phi$ a quantum channel, and $\Psi$ may also be taken to be a quantum channel in statement \ref{statement:left_inverse} (and hence, $\Phi$ is a reversible quantum channel).
\end{theorem}

Before continuing some comments on the theorem are in order. In statement \ref{statement:choi}, the norm $\norm{J(\Phi T_n)}_1$ appears, but this specific location of the transpose is an arbitrary notational choice. Using the definition of the Choi matrix and properties of the transpose, it may be verified that
\begin{equation}
	\norm{J(\Phi T_n)}_1 = \norm{J( T_m \Phi)}_1 = \norm{(T_m \otimes \id{n})(J(\Phi))}_1 = \norm{(\id{m} \otimes T_n)(J(\Phi))}_1 \label{equation:equiv_choi}
\end{equation}
for any linear map $\Phi : M_n \rightarrow M_m$. As per the discussion in Section \ref{section:background} regarding entanglement negativity, the expressions in Equation \eqref{equation:equiv_choi} can be roughly interpreted as a measure of how entangled $J(\Phi)$ is. With this interpretation, the characterization given in statement \ref{statement:choi} says that $\Phi$ is a complete trace-norm isometry if and only if its Choi matrix is maximally entangled (and has a particular normalization).

Statements \ref{statement:multiplication_full}, \ref{statement:multiplication_orthogonal}, and \ref{statement:multiplication_orthogonal_restricted} concern the map $\Phi$ preserving certain kinds of multiplication. The intuition for these statements comes from the explicit structure given in statement \ref{statement:structure}. However, in our proof, we show how they follow directly from the assumptions of $\Phi$ being either a complete or 2-trace-norm isometry. The benefit of these alternative proofs, which we give separately in Proposition \ref{proposition:subspace_multiplication} below, is that they work in more generality (i.e. in the proposition we only use that the domain is a subspace $V \subset M_n$), and so may be of independent interest. We also note that statement \ref{statement:multiplication_orthogonal_restricted} may seem oddly specific, but it is included for being specially suited for proving the theorem in the following section.

Lastly, while we give a complete proof of the above theorem, several equivalences may be deduced from \cite{li_isometries_2005}, whose title ``Isometries for Ky Fan Norms between Matrix Spaces'' is self-explanatory of its content. In particular, as a special case of the results therein, an explicit structural characterization of (not necessarily complete) trace-norm isometries taking $M_n$ into $M_m$ is given. From this, the explicit structure of complete trace-norm isometries may be deduced by refining this structure, and indeed, this refinement only requires the additional assumption that the map is a $2$-trace-norm isometry. Thus, the equivalence of statements \ref{statement:complete}, \ref{statement:2}, and \ref{statement:structure} may be viewed as a special case of the main theorem in \cite{li_isometries_2005}. Furthermore, the general technique of the proofs we give are in line with those of \cite{li_isometries_2005}, and with linear norm preserver problems more generally \cite{li_linear_2001}: translating between norm relations and algebraic relations for matrices. (See \cite{chan_isometries_2005} for a survey of results on isometries of matrix spaces for unitarily invariant norms.)

With this last comment, we begin the proof of Theorem \ref{theorem:completely_isometric} with the following equivalence between a trace-norm relation for a $2 \times 2$ block-matrix, and statements about how the blocks multiply.

\begin{proposition} \label{proposition:2x2_block}
For matrices $A,B,C,D \in M_n$, it holds that
\begin{equation}
	\biggnorm{\left(\begin{array}{cc}
		A & B \\
		C & D
	\end{array}\right)}_1 = \norm{A}_1+\norm{B}_1+\norm{C}_1+\norm{D}_1,
\end{equation}
if and only if
\begin{equation}
	A^*B = AC^* = D^*C = DB^* = 0.
\end{equation}
\begin{proof}
In \cite[Proposition 6]{puzzuoli_ancilla_2017} it was shown that, for any Hilbert-Schmidt orthogonal set of matrices $\{A_i\}_{i=1}^r$ -- all with the same dimensions, not necessarily square -- it holds that
\begin{equation}
	\biggnorm{\sum_{i=1}^r A_i}_1 = \sum_{i=1}^r \norm{A_i}_1,
\end{equation}
if and only if $A_i^*A_j = 0$ and $A_iA_j^* = 0$ for all $i \neq j$. The current proposition follows by application of this fact to the set $\{A \otimes E_{1,1},B \otimes E_{1,2}, C \otimes E_{2,1}, D \otimes E_{2,2}\} \subset M_n \otimes M_2$.
\end{proof}
\end{proposition}

Next, we prove a proposition containing some of the implications required for Theorem \ref{theorem:completely_isometric}, but in more generality. We note that the proof takes inspiration from multiplicative domain proofs for unital and completely positive linear maps on C$^*$-algebras (see \cite{choi_schwarz_1974} and \cite[Theorem 3.18]{paulsen_completely_2003}).


\begin{proposition} \label{proposition:subspace_multiplication}
Let $V \subset M_n$ be a subspace, and let $\Phi : V \rightarrow M_n$ be linear.
\begin{enumerate}
	\item If $\Phi$ is a 2-trace-norm isometry, then for $A,B \in V$ the following implications hold:
	\begin{itemize}
	\item $A^*B = 0 \implies \Phi(A)^*\Phi(B) = 0$, and
	\item $AB^* = 0 \implies \Phi(A)\Phi(B)^* = 0$.
	\end{itemize}
	\item If $\Phi$ is a complete trace-norm isometry, then for $A,B,C,D \in V$ the following implications hold:
	\begin{itemize}
	\item $A^*B = C^*D \implies \Phi(A)^*\Phi(B) = \Phi(C)^*\Phi(D)$, and
	\item $AB^* = CD^* \implies \Phi(A)\Phi(B)^* = \Phi(C)\Phi(D)^*$.
	\end{itemize}
\end{enumerate}
\begin{proof}
First, assume $\Phi$ is a 2-trace-norm isometry and let $A,B \in V$. Assuming $A^*B = 0$, we have
\begin{equation}
	\biggnorm{\left(\begin{array}{cc}
		\Phi(A) & \Phi(B) \\
		0 & 0
	\end{array}\right)}_1 =  \biggnorm{\left(\begin{array}{cc}
		A & B \\
		0 & 0
	\end{array}\right)}_1 
	= \norm{A}_1 + \norm{B}_1 = \norm{\Phi(A)}_1 + \norm{\Phi(B)}_1,
\end{equation}
where the second equality is by Proposition \ref{proposition:2x2_block}. Hence, also by Proposition \ref{proposition:2x2_block}, equality between the first and last expressions implies that $\Phi(A)^*\Phi(B) = 0$. Similarly, if $AB^* = 0$, then
\begin{equation}
	\biggnorm{\left(\begin{array}{cc}
		\Phi(A) & 0 \\
		\Phi(B) & 0
	\end{array}\right)}_1 =  \biggnorm{\left(\begin{array}{cc}
		A & 0 \\
		B & 0
	\end{array}\right)}_1 
	= \norm{A}_1 + \norm{B}_1 = \norm{\Phi(A)}_1 + \norm{\Phi(B)}_1,
\end{equation}
and so $\Phi(A)\Phi(B)^* = 0$.

Next, assume $\Phi$ is a complete trace-norm isometry on $V$, and let $A,B,C,D \in V$. If $A^*B = C^*D$, then
\begin{equation}
	\left(\begin{array}{cc}
		A & 0 \\
		-C & 0
	\end{array}\right)^*
	\left(\begin{array}{cc}
		B & 0 \\
		D & 0
	\end{array}\right) = 0.
\end{equation}
Under the assumption that $\Phi$ is completely trace-norm isometric, $\Phi\otimes \id{2}$ is a 2-trace-norm isometry, and so, by the 2-trace-norm isometry case, it holds that
\begin{equation}
	\left(\begin{array}{cc}
		\Phi(A) & 0 \\
		-\Phi(C) & 0
	\end{array}\right)^*
	\left(\begin{array}{cc}
		\Phi(B) & 0 \\
		\Phi(D) & 0
	\end{array}\right) = 0,
\end{equation}
giving $\Phi(A)^*\Phi(B) = \Phi(C)^*\Phi(D)$. Similarly, if $AB^* = CD^*$, then
\begin{equation}
	\left(\begin{array}{cc}
		A & -C \\
		0 & 0
	\end{array}\right)
	\left(\begin{array}{cc}
		B & D \\
		0 & 0
	\end{array}\right)^* = 0,
\end{equation}
and again the 2-trace-norm isometry case implies that
\begin{equation}
	\left(\begin{array}{cc}
		\Phi(A) & -\Phi(C) \\
		0 & 0
	\end{array}\right)
	\left(\begin{array}{cc}
		\Phi(B) & \Phi(D) \\
		0 & 0
	\end{array}\right)^* = 0,
\end{equation}
giving $\Phi(A)\Phi(B)^* = \Phi(C)\Phi(D)^*$.
\end{proof}
\end{proposition}

\begin{proof}[Proof of Theorem \ref{theorem:completely_isometric}]

We prove the implications in the diagram below.

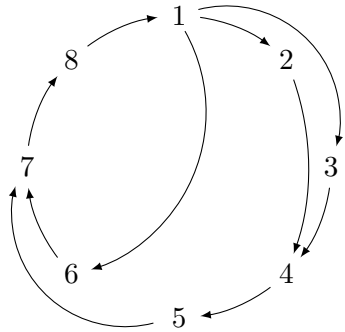
\begin{center}
\begin{tikzpicture}

\def \n {8}
\def \radius {2cm}
\def \margin {8} 

\foreach \s in {1,...,\n}
{
  \node at ({360/\n * (-\s + 1) + 90}:\radius) {$\s$};
}

 \draw[->,>=latex] ({360/\n *(1-1)-\margin + 90}:\radius) arc (360/\n *(1-1)-\margin+90:360/\n *(1-2)+\margin+90:\radius);
  \draw[->,>=latex] ({360/\n *(1-3)-\margin + 90}:\radius) arc (360/\n *(1-3)-\margin+90:360/\n *(1-4)+\margin+90:\radius);
\draw[->,>=latex] ({360/\n *(1-4)-\margin + 90}:\radius) arc (360/\n *(1-4)-\margin+90:360/\n *(1-5) +\margin+90:\radius);
\draw[->,>=latex] ({360/\n *(1-6)-\margin + 90}:\radius) arc (360/\n *(1-6)-\margin+90:360/\n *(1-7) +\margin+90:\radius);
\draw[->,>=latex] ({360/\n *(1-7)-\margin + 90}:\radius) arc (360/\n *(1-7)-\margin+90:360/\n *(1-8) +\margin+90:\radius);
\draw[->,>=latex] ({360/\n *(1-8)-\margin + 90}:\radius) arc (360/\n *(1-8)-\margin+90:360/\n *(1-9) +\margin+90:\radius);

 \draw[->,>=latex] ({(360/\n *(1-1)-\margin + 90)*1.0125}:{\radius*1.05}) arc (360/\n *(0.5)-\margin+90:360/\n *(-2.5)+\margin+90:{\radius*0.745});
 \draw[->,>=latex] ({(360/\n *(1-5)-\margin + 90)*1.0125}:{\radius*1.05}) arc (360/\n *(4.5)-\margin+90:360/\n *(-6.5)+\margin+90+360:{\radius*0.745});

\draw[->,>=latex] ({(360/\n *(1-2)-\margin + 90)*1.0125}:{\radius*0.95}) arc (360/\n *(-1.4)-\margin+90:360/\n *(-2.6)+\margin+90:{\radius*1.775});

 \draw[->,>=latex] ({87.5}:{\radius*(0.9)}) arc (360/\n *(-1.175)-\margin+90:360/\n *(-3.775)+\margin+90:{\radius*1.1});

\end{tikzpicture}
\end{center}

The implications appear in the order: \ref{statement:complete} $\Rightarrow$ \ref{statement:2}, \ref{statement:multiplication_full} $\Rightarrow$ \ref{statement:multiplication_orthogonal} $\Rightarrow$ \ref{statement:multiplication_orthogonal_restricted}, \ref{statement:complete} $\Rightarrow$ \ref{statement:choi}, \ref{statement:2} $\Rightarrow$ \ref{statement:multiplication_orthogonal}, \ref{statement:complete} $\Rightarrow$ \ref{statement:multiplication_full}, \ref{statement:left_inverse} $\Rightarrow$ \ref{statement:complete}, \ref{statement:choi} $\Rightarrow$ \ref{statement:structure}, \ref{statement:structure} $\Rightarrow$ \ref{statement:left_inverse}, and \ref{statement:multiplication_orthogonal_restricted} $\Rightarrow$ \ref{statement:structure}. All implications except \ref{statement:multiplication_orthogonal_restricted} $\Rightarrow$ \ref{statement:structure}, which is technically involved, follow essentially immediately from facts already given. The modified statements for the special case when $\Phi$ is positive are given before the proof of \ref{statement:multiplication_orthogonal_restricted} $\Rightarrow$ \ref{statement:structure}.

The implications that are immediate due to subsequent statements being logically weaker are \ref{statement:complete} $\Rightarrow$ \ref{statement:2} and \ref{statement:multiplication_full} $\Rightarrow$ \ref{statement:multiplication_orthogonal} $\Rightarrow$ \ref{statement:multiplication_orthogonal_restricted}. The implication \ref{statement:complete} $\Rightarrow$ \ref{statement:choi} follows from the norm relations 
\begin{equation}
	\norm{J(\id{n})}_1 = n\textnormal{ and }\norm{J(T_n)}_1 = n^2,
\end{equation} 
and $\Phi$ being a complete trace-norm isometry. The implications \ref{statement:2} $\Rightarrow$ \ref{statement:multiplication_orthogonal} and \ref{statement:complete} $\Rightarrow$ \ref{statement:multiplication_full} are both the content of Proposition \ref{proposition:subspace_multiplication}, and the implication \ref{statement:left_inverse} $\Rightarrow$ \ref{statement:complete} is straightforward to verify.

For \ref{statement:choi} $\Rightarrow$ \ref{statement:structure}, the norm values of statement \ref{statement:choi} imply by Theorem \ref{theorem:max_entangled_structure} that there exists a positive integer $r \leq m/n$, a density matrix $\sigma \in M_r$, and isometries $U,V : \complex^r \otimes \complex^n \rightarrow \complex^m$ for which
\begin{equation}
	\frac{1}{n}J(\Phi) = (U \otimes \I_n)(\sigma \otimes \tau_n)(V^* \otimes \I_n).
\end{equation} 
This is equivalent to the required form for $\Phi$.

For \ref{statement:structure} $\Rightarrow$ \ref{statement:left_inverse}, define $\Psi(Y) = \Tr_{M_r}(U^* Y V)$ for $Y \in M_m$, where $\Tr_{M_r}$ is the partial trace of $M_r$. Using the fact that the trace-norm is non-increasing under partial trace, it may be verified that $\Psi$ has the required properties.

For the special case when $\Phi$ is positive, return to the proof of the implication \ref{statement:choi} $\Rightarrow$ \ref{statement:structure}. As $\Phi$ is Hermiticity preserving, by Remark \ref{remark:pw_theorem} and Theorem \ref{theorem:max_entangled_structure}, there exists a Hermitian $H \in M_r$ with $\norm{H}_1 = 1$, and an isometry $U : \complex^n \otimes \complex^r \rightarrow \complex^m$ for which $\Phi(X) = U(X \otimes H)U^*$ for all $X \in M_n$. That $\Phi$ is positive implies $H \geq 0$, making $H$ a density matrix, and giving $\Phi$ the required form. To see that $\Psi$ may also be taken to be a quantum channel in statement \ref{statement:left_inverse}, we define $\Psi$ as as in the proof of \ref{statement:structure} $\Rightarrow$ \ref{statement:left_inverse} with a slight modification. Fix a density matrix $\eta \in M_n$, and set $\Psi(Y) = \Tr_{M_r}(U^* Y U) + \Tr((\I_m - UU^*)Y)\eta$ for all $Y \in M_m$. It is routine to verify that $\Psi$ is a quantum channel and that $\Psi \Phi = \id{n}$. 

Lastly, we show \ref{statement:multiplication_orthogonal_restricted} $\Rightarrow$ \ref{statement:structure}. We will use the assumption in statement \ref{statement:multiplication_orthogonal_restricted} to build further facts about how outputs of $\Phi$ on rank-1 matrices multiply, which we break into a series of claims.

\textbf{Claim 1.} For unit vectors $x_1,x_2,y \in \complex^n$ with $\ip{x_1}{x_2} = 0$, it holds that 
\begin{equation}
	\Phi(x_1y^*)^*\Phi(x_1y^*) = \Phi(x_2y^*)^*\Phi(x_2y^*)\textnormal{ and }\Phi(yx_1^*)\Phi(yx_1^*)^* = \Phi(yx_2^*)\Phi(yx_2^*)^*.
\end{equation}
To see the first equality, note that $x_1+x_2 \perp x_1-x_2$, and so
\begin{align}
	0&=\Phi((x_1+x_2)y^*)^*\Phi((x_1-x_2)y^*) \\
	&=\Phi(x_1y^*)^*\Phi(x_1y^*) - \Phi(x_1y^*)^*\Phi(x_2y^*) + \Phi(x_2y^*)^*\Phi(x_1y^*) - \Phi(x_2y^*)^*\Phi(x_2y^*) \label{equation:first_expansion}\\
	&= \Phi(x_1y^*)^*\Phi(x_1y^*)- \Phi(x_2y^*)^*\Phi(x_2y^*),
\end{align}
where the second and third term in Equation \eqref{equation:first_expansion} are $0$ by application of statement \ref{statement:multiplication_orthogonal_restricted}. This gives the desired equality, and it follows similarly that $\Phi(yx_1^*)\Phi(yx_1^*)^* = \Phi(yx_2^*)\Phi(yx_2^*)^*$.

\textbf{Claim 2.} For any $x_1,x_2,y_1,y_2 \in \complex^n$ with $\ip{x_1}{x_2} = \ip{y_1}{y_2} = 0$, it holds that
\begin{equation}
	\Phi(x_1y_1^*)^*\Phi(x_2y_2^*) = 0\textnormal{ and } \Phi(x_1y_1^*)\Phi(x_2y_2^*)^* = 0.
\end{equation}
For the first equality, assuming without loss of generality that $x_1, x_2, y_1,$ and $y_2$ have unit length, Claim 1 gives that $\Phi(x_1y_1^*)\Phi(x_1y_1^*)^* = \Phi(x_1y_2^*)\Phi(x_1y_2^*)^*$, and statement \ref{statement:multiplication_orthogonal_restricted} gives that $\Phi(x_1y_2^*)^* \Phi(x_2y_2^*) = 0$. These equalities imply the range relations
\begin{equation}
	\ran(\Phi(x_1y_1^*)) = \ran(\Phi(x_1y_2^*)) \perp \ran(\Phi(x_2y_2^*)),
\end{equation}
which imply $\Phi(x_1y_1^*)^*\Phi(x_2y_2^*) = 0$. It similarly follows that $\Phi(x_1y_1^*)\Phi(x_2y_2^*)^* = 0$.

\textbf{Claim 3.} For unit vectors $x_1,x_2,y_1,y_2 \in \complex^n$ with $\ip{x_1}{x_2} = \ip{y_1}{y_2} = 0$, it holds that
\begin{equation}
	\Phi(x_1y_1^*)^*\Phi(x_1y_2^*) = \Phi(x_2y_1^*)^*\Phi(x_2y_2^*)\textnormal{, and } \Phi(x_1y_1^*)\Phi(x_2y_1^*)^* = \Phi(x_1y_2^*)\Phi(x_2y_2^*)^*.
\end{equation}
For the first equality, $x_1+x_2 \perp x_1-x_2$, so Claim 2 implies
\begin{align}
	0 &= \Phi((x_1+x_2)y_1^*)^*\Phi((x_1-x_2)y_2^*) \\
	&= \Phi(x_1y_1^*)^*\Phi(x_1y_2^*) - \Phi(x_1y_1^*)^*\Phi(x_2y_2^*) + \Phi(x_2y_1^*)^*\Phi(x_1y_2^*) - \Phi(x_2y_1^*)^*\Phi(x_2y_2^*) \\
	&= \Phi(x_1y_1^*)^*\Phi(x_1y_2^*) - \Phi(x_2y_1^*)^*\Phi(x_2y_2^*),
\end{align}
where we have used Claim 2 to determine that $\Phi(x_1y_1^*)^*\Phi(x_2y_2^*) = \Phi(x_2y_1^*)^*\Phi(x_1y_2^*) = 0$. It may be similarly reasoned that $\Phi(x_1y_1^*)\Phi(x_2y_1^*)^* = \Phi(x_1y_2^*)\Phi(x_2y_2^*)^*$.

We now use Claims 1 to 3 to construct the explicit structure of $\Phi$ by examining how it acts on elementary matrices. By Claim 1 it holds that 
\begin{equation}
	\Phi(E_{1,1})^*\Phi(E_{1,1}) = \Phi(E_{a,1})^*\Phi(E_{a,1})\textnormal{ and } \Phi(E_{1,1})\Phi(E_{1,1})^* = \Phi(E_{1,b})\Phi(E_{1,b})^*
\end{equation}
for all $1 \leq a,b \leq n$, and hence there exists partial isometries
\begin{equation}
	U_a : \ran(\Phi(E_{1,1})) \rightarrow \ran(\Phi(E_{a,1}))\textnormal{, and } V_b : \ran(\Phi(E_{1,1})^*) \rightarrow \ran(\Phi(E_{1,b})^*)
\end{equation}
for which $\Phi(E_{a,1}) = U_a \Phi(E_{1,1})$ and $\Phi(E_{1,b}) = \Phi(E_{1,1})V_b^*$ (where $U_1$ and $V_1$ may be taken to be the orthogonal projections onto $\ran(\Phi(E_{1,1}))$ and $\ran(\Phi(E_{1,1})^*)$ respectively). Note as well that, statement \ref{statement:multiplication_orthogonal_restricted} gives that $\Phi(E_{a,1})^*\Phi(E_{a',1}) = \Phi(E_{1,b})\Phi(E_{1,b'})^* = 0$ for $a \neq a'$ and $b \neq b'$, and hence the sets of partial isometries $\{U_a\}_{a=1}^n, \{V_b\}_{b=1}^n$ have mutually orthogonal ranges.

Next, we claim that for all $1 \leq a,b \leq n$ it holds that $\Phi(E_{a,b}) = U_a\Phi(E_{1,1})V_b^*$. In the previous claim this fact is established when at least one of $a$ and $b$ is 1, so we may assume that both $a,b \geq 2$. In this case, Claim 3 implies that
\begin{equation}
	\Phi(E_{1,1})^*\Phi(E_{1,b}) = \Phi(E_{a,1})^*\Phi(E_{a,b}),
\end{equation}
and so
\begin{equation}
	\Phi(E_{1,1})^*\Phi(E_{1,1})V_b^* = \Phi(E_{1,1})^*U_a^*\Phi(E_{a,b}).
\end{equation}
As $\ran(U_a^*) = \ran(\Phi(E_{1,1}))$ we may cancel the $\Phi(E_{1,1})^*$ from the left-side of the above equation to get that $\Phi(E_{1,1})V_b^* = U_a^*\Phi(E_{a,b})$ (alternatively, we may multiply on the left by the pseudo-inverse of $\Phi(E_{1,1})^*$). Finally, we have $\ran(U_a) = \ran(E_{a,b})$, as $\Phi(E_{a,b})^*\Phi(E_{a,b}) = \Phi(E_{a,1})^*\Phi(E_{a,1})$, and so
\begin{equation}
	\Phi(E_{a,b}) = U_aU_a^*\Phi(E_{a,b}) = U_a\Phi(E_{1,1})V_b^*,
\end{equation}
as required.

The last step is to show that the structure we have just deduced for $\Phi$ is the same as that in statement \ref{statement:structure}. Let $\Phi(E_{1,1}) = \sum_{i=1}^r s_i x_iy_i^*$ be a singular value decomposition. Define $\sigma = \sum_{i=1}^r s_i E_{i,i} \in M_r$, which is clearly positive, and define matrices $U,V : \complex^n \otimes \complex^r \rightarrow \complex^m$ to act as
\begin{equation}
	U(e_a \otimes e_i) = U_ax_i\textnormal{, and } V(e_b \otimes e_j) = V_by_j.
\end{equation}
We may verify that these are in fact isometries:
\begin{equation}
	\ip{U(e_b \otimes e_j)}{U(e_a \otimes e_i)} = \ip{U_b x_j}{U_a x_i} = \delta_{a,b} \ip{x_j}{x_i} = \delta_{a,b} \delta_{i,j},
\end{equation}
where we have used that $U_a$ and $U_b$ have orthogonal ranges for $a \neq b$. Hence, $U$ is an isometry as it sends an orthonormal basis to an orthonormal set. The same proof shows that $V$ is an isometry. Finally, we have that
\begin{equation}
	\Phi(E_{a,b}) = U_a\Phi(E_{1,1})V_b^* = \sum_{i=1}^r s_i U_a x_i y_i^* V_b^* = \sum_{i=1}^r s_i U(E_{a,b} \otimes E_{i,i})V^* = U(E_{a,b} \otimes \sigma)V^*.
\end{equation}

Hence, $\Phi$ has the desired form, and the last thing we need is that $\Tr(\sigma) = 1$. The final assumption is that there exists $X \in M_n\setminus \{0\}$ with $\norm{\Phi(X)}_1 = \norm{X}_1$. This gives
\begin{equation}
	\norm{X}_1 = \norm{\Phi(X)}_1 = \norm{U(X \otimes \sigma)V^*}_1 = \norm{X}_1 \Tr(\sigma),
\end{equation}
and hence $\Tr(\sigma) = 1$ as desired.
\end{proof}

\begin{remark} \label{remark:herm_trace_norm_pres}
Consider an additional special case of Theorem \ref{theorem:completely_isometric} when $\Phi$ is Hermiticity preserving. As in the proof of the case when $\Phi$ is positive, there exists a positive integer $r \leq m/n$, a Hermitian $H \in M_r$, and an isometry $U : \complex^n \otimes \complex^r \rightarrow \complex^m$ for which $\Phi(X) = U(X \otimes H)U^*$ for all $X \in M_n$. If $m < 2n$, then necessarily $r = 1$ and hence $H = \pm 1$. It follows that either $\Phi$ or $-\Phi$ is a reversible quantum channel. If $m \geq 2n$, then by considering the Hahn decomposition of $H$, one may verify that this form is equivalent to the statement that there exists reversible quantum channels $\Phi_0,\Phi_1 : M_n \rightarrow M_m$ with orthogonal ranges and a number $r \in [0,1]$ for which
\begin{equation}
	\Phi = r \Phi_0 - (1-r)\Phi_1.
\end{equation}
\end{remark}

\section{Characterization of linear maps whose multiplicity maps have maximal norm} \label{section:max_norm_inflation}


We now prove the main result. As mentioned in the introduction, the inequality in Equation \eqref{equation:induced_norm_inequality} below is known in more generality in C$^*$-algebras \cite[Exercise 3.10]{paulsen_completely_2003}.

\begin{theorem} \label{theorem:maximum_norm_gap}
Let $\Phi : M_n \rightarrow M_m$ be linear with $\norm{\Phi}_1 = 1$. It holds that
\begin{equation}
	\norm{\Phi \otimes \id{k}}_1 \leq k, \label{equation:induced_norm_inequality}
\end{equation}
with equality if and only if $n,m \geq k$, and for any pair of unit vectors $u,v \in \complex^n \otimes \complex^k$ satisfying
\begin{equation}
	\norm{(\Phi \otimes \id{k})(uv^*)}_1 = k
\end{equation}
(of which at least one such pair must exist), the following statements hold:
\begin{enumerate}
	\item $u$ and $v$ are maximally entangled; i.e. they decompose as
	\begin{equation}
		u = \sqrt{\frac{1}{k}} \sum_{a=1}^k u_a \otimes e_a\textnormal{ and } v = \sqrt{\frac{1}{k}} \sum_{b=1}^k v_b \otimes e_b
	\end{equation}
	for orthonormal sets $\{u_a\}_{a =1}^k,\{v_b \}_{b=1}^k \subset \complex^n$.
	\item Defining isometries $U,V : \complex^k \rightarrow \complex^n$ as 
	\begin{equation}
		U = \sum_{a=1}^k u_ae_a^*\textnormal{ and }V = \sum_{b=1}^k v_be_b^*, \label{equation:isometry_definition}
	\end{equation}
	there exists a complete trace-norm isometry $\Psi : T_n(UM_kV^*) \rightarrow M_m$, where 
	\begin{equation}
		T_n(UM_kV^*) = \{(UXV^*)^\t : X \in M_k\},
	\end{equation} 
	for which $\Phi(X) = \Psi(X^\t)$ for all $X \in UM_kV^*$.
\end{enumerate}
\begin{proof}

Letting $u,v \in \complex^n \otimes \complex^k$ be unit vectors, we will first show that 
\begin{equation}
	\norm{(\Phi \otimes \id{k})(uv^*)}_1 \leq k,
\end{equation} 
which will prove Equation \eqref{equation:induced_norm_inequality}. We may assume without loss of generality that $u$ and $v$ have decompositions of the form
\begin{equation}
	u = \sum_{a=1}^r \alpha_a u_a \otimes e_a\textnormal{, and } v= \sum_{b=1}^r \beta_b v_b \otimes e_b, \label{equation:general_unit_vector_decomp}
\end{equation}
for $r \leq \min(k,n)$, unit vectors $\alpha,\beta \in \complex^r$ with non-negative entries, and orthonormal sets $\{u_a\}_{a=1}^r,\{v_b\}_{b=1}^r \subset \complex^n$. We have
\begin{align}
\norm{(\Phi \otimes \id{k})(uv^*)}_1 & = \biggnorm{\sum_{a,b=1}^r \alpha_a \beta_b \Phi(u_av_b^*) \otimes E_{a,b}}_1 \\
	& \leq \sum_{a,b=1}^r \alpha_a \beta_b \norm{\Phi(u_av_b^*)}_1 \label{equation:trace_norm_triangle} \\
	& \leq \sum_{a,b=1}^r \alpha_a \beta_b \norm{u_av_b^*}_1 \label{equation:termwise_norm}\\
	& = \sum_{a,b=1}^r \alpha_a\beta_b \\
	& = \ip{1_r}{\alpha}\ip{1_r}{\beta} \\
	& \leq \norm{1_r}^2 \norm{\alpha} \norm{\beta} \label{equation:cauchy_schwarz} \\
	& = r\\
	& \leq k
\end{align}
where $1_r \in \complex^r$ is the vector of all ones. Hence, it holds that $\norm{\Phi \otimes \id{k}}_1 \leq k$.

We now examine equality conditions. Suppose that $\norm{(\Phi \otimes \id{k})(uv^*)}_1 = k$ for unit vectors $u,v \in \complex^n \otimes \complex^k$ with decompositions as in Equation (\ref{equation:general_unit_vector_decomp}). First, we may conclude that $r = k$, and hence $k \leq n$. Furthermore, equality in the application of Cauchy-Schwarz in Equation (\ref{equation:cauchy_schwarz}) implies that $\alpha = \beta = \sqrt{\frac{1}{k}}1_k$, and so $u$ and $v$ are maximally entangled.

Thus, defining the isometries $U,V : \complex^k \rightarrow \complex^n$ as in Equation \eqref{equation:isometry_definition}, it holds that
\begin{equation}
	\biggnorm{\sum_{a,b=1}^k \Phi(UE_{a,b}V^*) \otimes E_{a,b}}_1 = k^2,
\end{equation}
and this is equivalent to the more general fact that
\begin{equation}
	\biggnorm{\sum_{a,b=1}^k \Phi(x_ay_b^*) \otimes E_{a,b}}_1 = k^2,
\end{equation}
for any orthonormal bases $\{x_a\}_{a=1}^k \subset \ran(U)$ and $\{y_b\}_{b=1}^k \subset \ran(V)$. As $\norm{\Phi}_1 = 1$, the above implies that $\norm{\Phi(x_ay_b^*)}_1 = 1$ for all $1 \leq a,b \leq k$. By looking at $2 \times 2$ block-sub-matrices, it also implies that, for any unit vectors $x_1, x_2 \in \ran(U)$ and $y_1, y_2 \in \ran(V)$ with $\ip{x_1}{x_2} = \ip{y_1}{y_2} = 0$, it holds that 
\begin{equation}
	\biggnorm{ \left(\begin{array}{cc}
				\Phi(x_1y_1^*) & \Phi(x_1y_2^*) \\
				\Phi(x_2y_1^*) & \Phi(x_2y_2^*)
				\end{array}\right)}_1 = 4.
\end{equation}
As each block has trace-norm $1$, the above $2 \times 2$ block matrix has trace-norm equal to the sum of the trace-norms of the blocks. Proposition \ref{proposition:2x2_block} then implies the relations
\begin{equation}
	\Phi(x_1y_1^*)^*\Phi(x_1y_2^*) = \Phi(x_1y_1^*)\Phi(x_2y_1^*)^* = 0
\end{equation}
for any $x_1, x_2 \in \ran(U)$ and $y_1,y_2 \in \ran(V)$ with $\ip{x_1}{x_2} = \ip{y_1}{y_2} = 0$. 

This may be written in a more suggestive way: for $A,B \in UM_kV^*$ rank-1, the following implications hold:
\begin{enumerate}[(i)]
	\item $A^*B = 0$ and $A^*A = B^*B \implies \Phi(A)\Phi(B)^* = 0$, and
	\item $AB^* = 0$ and $AA^* = BB^* \implies \Phi(A)^*\Phi(B) = 0$.
\end{enumerate}
These implications are very similar to statement \ref{statement:multiplication_orthogonal_restricted} in Theorem \ref{theorem:completely_isometric}, but the adjoints appear in different locations. We may remedy this by defining $\Psi : T_n(U M_k V^*) \rightarrow M_m$ as $\Psi = \Phi T_n$, where $T_n$ is the transpose. We claim that, for $A,B \in T_n(UM_kV^*)$ rank-1, the following implications hold:
\begin{enumerate}[(a)]
	\item $A^*B = 0$ and $A^*A = B^*B \implies \Psi(A)^*\Psi(B) = 0$, and
	\item $AB^* = 0$ and $AA^* = BB^* \implies \Psi(A)\Psi(B)^* = 0$.
\end{enumerate}
We will prove that (ii) $\Rightarrow$ (a), with (i) $\Rightarrow$ (b) being similar. Let $A^\t,B^\t \in T_n(UM_kV^*)$ be rank-1. The statements $(A^\t)^*B^\t = 0$ and $(A^\t)^*A^\t = (B^\t)^*B^\t$ are equivalent to $AB^* = 0$ and $AA^* = BB^*$, so (ii) implies that $\Phi(A)^*\Phi(B) = 0$, which is in turn equivalent to $\Psi(A^\t)^*\Psi(B^\t) = 0$. 

Thus, by Theorem \ref{theorem:completely_isometric}, $\Psi = \Phi T_n$ is a complete trace-norm isometry on $T_n(UM_kV^*)$, as required.\footnote{Note that Theorem \ref{theorem:completely_isometric} as stated only applies to maps whose domain is all of $M_n$. Here, the domain of $\Psi$ is $T_n(UM_k V^*) = \overline{V}M_kU^\t \subset M_n$, so technically we are applying Theorem \ref{theorem:completely_isometric} to conclude that the linear map $X \mapsto \Psi(V^\t X \overline{U})$ is a complete trace-norm isometry on $M_k$. However, this is equivalent to $\Psi$ being a complete trace-norm isometry on $T_n(UM_kV^*)$.}
\end{proof}
\end{theorem}

As a corollary to Theorems \ref{theorem:completely_isometric} and \ref{theorem:maximum_norm_gap}, we provide two characterizations of the set of linear maps $\Phi : M_n \rightarrow M_m$ satisfying $\vertiii{\Phi}_1 = n \norm{\Phi}_1$. 

\begin{corollary} \label{corollary:maximal_norm_gap}
Let $\Phi : M_n \rightarrow M_m$ be linear with $m \geq n$. The following are equivalent. 
\begin{enumerate}
	\item $\norm{\Phi}_1 = 1$ and $\vertiii{\Phi}_1 = n$.
	\item $\norm{J(\Phi)}_1 = n^2$ and $\norm{J(\Phi T_n)}_1 = n$.
	\item There exists a complete trace-norm isometry $\Psi : M_n \rightarrow M_m$ for which $\Phi=\Psi T_n$.
\end{enumerate} 
In the above, if $\Phi$ is Hermiticity preserving so is $\Psi$, and if $\Phi$ is positive then so is $\Psi$ (and hence is a reversible quantum channel).
\begin{proof}
It is immediate that $3 \Rightarrow 1$ and $3 \Rightarrow 2$. That $1 \Rightarrow 3$ is given by Theorem \ref{theorem:maximum_norm_gap}, and that $2 \Rightarrow 3$ is given by Theorem \ref{theorem:completely_isometric}. For the special cases, since $\Psi = \Phi T_n$, if $\Phi$ is Hermiticity preserving so is $\Psi$, as it is a composition of Hermiticity preserving maps. The same logic applies if $\Phi$ is positive; with $\Psi$ being a reversible quantum channel following from the positive case of Theorem \ref{theorem:completely_isometric}.
\end{proof}
\end{corollary}

\section{A uniqueness result for the Werner-Holevo channels in single-shot quantum channel discrimination} \label{section:werner_holevo}

A fundamental operational task in quantum information is to determine which quantum channel, from a set of possible channels, is acting on a system. The simplest version of this task is \emph{single-shot quantum channel discrimination}, where the goal is to determine which of two channels is acting on a system given only a single use. Various aspects of this task have been extensively studied (see e.g. \cite{Kitaev97,KitaevSV02,Sacchi05,Sacchi05b,Rosgen08,Watrous08,piani_all_2009,jencova_conditions_2016,puzzuoli_ancilla_2017}). For completeness, we give a description of the task below. Proofs of all facts summarized may be found in \cite[Chapter 3]{watrous_quantum_2017}.

Single-shot quantum channel discrimination may be formulated as a single-player game specified by a triple $(\lambda, \Gamma_0, \Gamma_1)$, where $\lambda \in [0,1]$ and $\Gamma_0,\Gamma_1 : M_n \rightarrow M_m$ are quantum channels. In the game, the player knows a description of the triple $(\lambda, \Gamma_0, \Gamma_1)$, and the game proceeds as follows:
\begin{enumerate}
	\item The referee samples a bit $\alpha \in \{0,1\}$ with probability $p(0) =\lambda$, $p(1)= 1-\lambda$.
	\item The player is given a single use of the channel $\Gamma_\alpha$; i.e. the player gives a quantum state on the input system of their choice to the referee, who then returns the output of $\Gamma_\alpha$.
	\item The player must guess $\alpha$ (after say, making a measurement on the output). 
\end{enumerate}
The goal of the player is to maximize the probability that they guess $\alpha$ correctly.

Ultimately, all the player can do is prepare an input state to $\Gamma_\alpha$ then try to discriminate the two possible outputs. In the most unconstrained version of the game, the player is free to use an auxilliary system; i.e. they can prepare a bipartite quantum state $\rho \in M_n \otimes M_k$ then discriminate the outputs $(\Gamma_\alpha \otimes \id{k})(\rho)$. Due to the Holevo-Helstrom theorem for single-shot quantum state discrimination \cite{helstrom_detection_1967,holevo_analog_1972}, the optimal probability of success given a choice of state $\rho \in M_n \otimes M_k$ is
\begin{equation}
	\frac{1}{2} + \frac{1}{2}\bignorm{\lambda \big(\Gamma_0 \otimes \id{k}\big)(\rho) - (1-\lambda)\big(\Gamma_1 \otimes \id{k}\big)(\rho)}_1.
\end{equation}
Thus, for a fixed auxiliary system of dimension $k$, the optimal success probability of winning the game is given by the optimization of the above expression over density matrices, which reduces to
\begin{equation}
	\frac{1}{2} + \frac{1}{2}\norm{\lambda \Gamma_0 \otimes \id{k} - (1-\lambda)\Gamma_1 \otimes \id{k}}_{1,H}, \label{equation:optimal_k}
\end{equation}
where, for $\Psi : M_n \rightarrow M_m$,
\begin{equation}
	\norm{\Psi}_{1,H} = \max \{ \norm{\Psi(H)}_1 : H \in M_n, \norm{H}_1 = 1, H=H^*\}.
\end{equation}
Hence, the optimal value over unconstrained strategies, which amounts to optimizing Equation \eqref{equation:optimal_k} over $k \geq 1$, is given by 
\begin{equation}
	\frac{1}{2} + \frac{1}{2}\vertiii{\lambda \Gamma_0 - (1-\lambda)\Gamma_1 }_1. \label{equation:optimal}
\end{equation}

A natural question to ask in this setting is: What channel discrimination games (with system input size $n$) have the maximum possible gap between the optimal performance with and without entanglement? By Equations \eqref{equation:optimal_k} and \eqref{equation:optimal}, this amounts to characterizing the games $(\lambda, \Gamma_0, \Gamma_1)$, where $\Gamma_0, \Gamma_1 : M_n \rightarrow M_m$, with maximal gap between the norms
\begin{equation}
	\norm{\lambda \Gamma_0 - (1-\lambda) \Gamma_1}_{1,H}\:\:\textnormal{ and }\:\:\vertiii{\lambda \Gamma_0 - (1-\lambda) \Gamma_1}_1.
\end{equation}
In this section we make partial progress towards answering this question. We apply the results of the previous section to characterize the games (with input space $M_n$) that have maximal gap between the norms $\norm{\lambda \Gamma_0 - (1-\lambda) \Gamma_1}_1$ and $\vertiii{\lambda \Gamma_0 - (1-\lambda) \Gamma_1}_1$. As it generically holds that
\begin{equation}
	0 \leq \norm{\lambda \Gamma_0 - (1-\lambda) \Gamma_1}_1 \leq  \vertiii{\lambda \Gamma_0 - (1-\lambda) \Gamma_1}_1 \leq 1,
\end{equation}
and $\vertiii{\lambda \Gamma_0 - (1-\lambda) \Gamma_1}_1 \leq n\norm{\lambda \Gamma_0 - (1-\lambda) \Gamma_1}_1$, the maximum possible gap occurs when
\begin{equation}
	1 = \vertiii{\lambda \Gamma_0 - (1-\lambda) \Gamma_1}_1 = n\norm{\lambda \Gamma_0 - (1-\lambda) \Gamma_1}_1.\label{equation:game_max_norm_gap}
\end{equation}
Operationally, the above relations say that the game can be won with certainty using arbitrary entanglement, but is hard to win without entanglement, with the upper bound on unentangled performance given by $\norm{\lambda \Gamma_0 - (1-\lambda) \Gamma_1}_1$ being as small as possible given that the game can be won with certainty.

As detailed in \cite[Example 3.36]{watrous_quantum_2017}, the Werner-Holevo channels \cite{werner_counterexample_2002} provide a well-known family of channel discrimination games satisfying Equation \eqref{equation:game_max_norm_gap}. For $n \geq 2$, we denote them as $\Phi_n^{(0)},\Phi_n^{(1)} : M_n \rightarrow M_n$, and they are defined to act as
\begin{equation}
	\Phi_n^{(0)}(X) = \frac{1}{n+1}(\Tr(X)\I_n + X^\t)\textnormal{, and }\Phi_n^{(1)}(X) = \frac{1}{n-1}(\Tr(X)\I_n - X^\t) \label{equation:werner_holevo_definition}
\end{equation}
for all $X \in M_n$. For the probability $\lambda_n = \frac{n+1}{2n}$, it holds that
\begin{equation}
	\lambda_n \Phi_n^{(0)} - (1-\lambda_n) \Phi_n^{(1)} = \frac{1}{n}T_n, \label{equation:transpose_decomp}
\end{equation}
and as such
\begin{equation}
	1 = \Bigvertiii{\lambda_n \Phi_n^{(0)} - (1-\lambda_n) \Phi_n^{(1)}}_1 = n \Bignorm{\lambda_n \Phi_n^{(0)} - (1-\lambda_n) \Phi_n^{(1)}}_1. \label{equation:channel_descrim_norms}
\end{equation}
Thus, for $n \geq 2$ the triple $\big(\lambda_n, \Phi_n^{(0)}, \Phi_n^{(1)}\big)$ is an example of a channel discrimination game satisfying Equation \eqref{equation:game_max_norm_gap}.\footnote{While the Werner-Holevo channels were originally introduced in \cite{werner_counterexample_2002} for other reasons, in the setting of quantum channel discrimination, they may be intuitively viewed as a way of smuggling the transpose into the problem, with the properties of the resulting game being inherited from norm properties of the transpose.} We now apply the result of the previous section to show that the game specified by the triple $\big(\lambda_n,\Phi_n^{(0)},\Phi_n^{(1)}\big)$ is in some sense the unique game (with channels having domain $M_n$) satisfying Equation \eqref{equation:game_max_norm_gap}.

\begin{theorem} \label{theorem:werner_holevo_unique}
Let $m \geq n$, $\Gamma_0,\Gamma_1 : M_n \rightarrow M_m$ be quantum channels, $\lambda \in (0,1)$ be a probability, and let $\Phi_n^{(0)},\Phi_n^{(1)} : M_n \rightarrow M_n$ be the Werner-Holevo channels as given in Equation \eqref{equation:werner_holevo_definition}. It holds that
\begin{equation}
	1 = \vertiii{\lambda \Gamma_0 - (1-\lambda) \Gamma_1}_1 = n\norm{\lambda \Gamma_0 - (1-\lambda) \Gamma_1}_1 \label{equation:maximal_gap_channel_difference}
\end{equation}
if and only if:
\begin{itemize}
\item For $m < 2n$, there exists a reversible quantum channel $\Psi : M_n \rightarrow M_m$ for which either $(\lambda, \Gamma_0,\Gamma_1) = (\lambda_n,\Psi \Phi_n^{(0)}, \Psi \Phi_n^{(1)})$, or $(\lambda, \Gamma_0,\Gamma_1) = (1-\lambda_n,\Psi \Phi_n^{(1)}, \Psi \Phi_n^{(0)})$.
\item For $m \geq 2n$, there exists $r \in [0,1]$ and two reversible channels $\Psi_0,\Psi_1 : M_n \rightarrow M_m$ with orthogonal ranges for which $\lambda = r \lambda_n + (1-r)(1-\lambda_n)$, and
\begin{equation}
	\lambda \Gamma_0 = r \lambda_n \Psi_0 \Phi_n^{(0)} + (1-r)(1 - \lambda_n) \Psi_1 \Phi_n^{(1)}, \label{equation:gamma_0}
\end{equation}
and
\begin{equation}
	(1-\lambda) \Gamma_1 = r (1-\lambda_n) \Psi_0 \Phi_n^{(1)} + (1-r)\lambda_n \Psi_1 \Phi_n^{(0)}. \label{equation:gamma_1}
\end{equation}
\end{itemize}
\end{theorem}

\begin{remark}
Theorem \ref{theorem:werner_holevo_unique} may be interpreted as saying that the game $(\lambda_n, \Phi_n^{(0)}, \Phi_n^{(1)})$ uniquely satisfies Equation \eqref{equation:maximal_gap_channel_difference} in the following sense: Any game $(\lambda, \Gamma_0, \Gamma_1)$, whose channels have domain $M_n$ and satisfy Equation \eqref{equation:maximal_gap_channel_difference}, is constructed out of, and is reducible by the player to, the game $(\lambda_n, \Phi_n^{(0)}, \Phi_n^{(1)})$ in a way that perfectly preserves success probabilities. Indeed, mathematically, one can check that
\begin{align}
	\|\lambda (\Gamma_0 \otimes \id{k})(X) - &(1-\lambda)(\Gamma_1 \otimes \id{k})(X)\|_1 \\&= \bignorm{\lambda_n \big(\Phi_n^{(0)} \otimes \id{k}\big)(X) - (1-\lambda_n)\big(\Phi_n^{(1)} \otimes \id{k}\big)(X)}_1 \label{equation:equal_probability}
\end{align}
for all integers $k\geq 1$ and matrices $X \in M_n \otimes M_k$. Operationally, the construction/reduction of such games $(\lambda, \Gamma_0, \Gamma_1)$ in terms of $(\lambda_n, \Phi_n^{(0)}, \Phi_n^{(1)})$ goes as follows.
\begin{itemize}
	\item For the case $m < 2n$, the construction and reduction are natural; up to a reversible quantum channel (which the player can undo) and a relabeling of the channels (which the player knows), the game $(\lambda, \Gamma_0, \Gamma_1)$ is exactly the game $\big(\lambda_n, \Phi_n^{(0)},\Phi_n^{(1)}\big)$.
	\item For the case $m \geq 2n$, the relation between $(\lambda, \Gamma_0, \Gamma_1)$ and $\big(\lambda_n, \Phi_n^{(0)}, \Phi_n^{(1)}\big)$ is less clear, though it can be thought of as a convex combination of relabelings of the game $\big(\lambda_n, \Phi_n^{(0)}, \Phi_n^{(1)}\big)$, where the player is able to detect which labeling is being used. Specifically, with probability $r$, $\Gamma_0$ acts as $\Phi_n^{(0)}$ and $\Gamma_1$ acts as $\Phi_n^{(1)}$, and with probability $(1-r)$ the labels are reversed. As $\Psi_0$ and $\Psi_1$ have orthogonal ranges, the player is able to measure which labelling is being used without disturbance. Once this is done, the situation from the players perspective is now the same as in the case $m < 2n$, and they may act accordingly.
\end{itemize}
\end{remark}

Before proving Theorem \ref{theorem:werner_holevo_unique}, we prove a lemma regarding the uniqueness of certain decompositions of Hermiticity preserving maps into differences of completely positive maps.

\begin{lemma} \label{lemma:CP_difference}
Let $\Phi : M_n \rightarrow M_m$ be Hermiticity preserving, $\Psi_0,\Psi_1 :M_n \rightarrow M_m$ be completely positive and satisfy
\begin{equation}
	\Phi = \Psi_0 - \Psi_1\textnormal{ and } \vertiii{\Phi}_1 = \vertiii{\Psi_0}_1 + \vertiii{\Psi_1}_1, \label{equation:decomposition_properties}
\end{equation}
and let $u \in \complex^n \otimes \complex^n$ be a unit vector satisfying $\vertiii{\Phi}_1 = \norm{(\Phi \otimes \id{n})(uu^*)}_1$. It follows that
\begin{equation}
	\vertiii{\Psi_0}_1 = \norm{(\Psi_0 \otimes \id{n})(uu^*)}_1\textnormal{ and }\vertiii{\Psi_1}_1 = \bignorm{(\Psi_1 \otimes \id{n})(uu^*)}_1,
\end{equation}
and for any other completely positive maps $\Psi_0', \Psi_1' :M_n \rightarrow M_m$ satisfying the conditions in Equation \eqref{equation:decomposition_properties},
\begin{equation}
		(\Psi_0' \otimes \id{n})(uu^*) = (\Psi_0 \otimes \id{n})(uu^*)\textnormal{ and }(\Psi_1' \otimes \id{n})(uu^*) = (\Psi_1 \otimes \id{n})(uu^*). \label{equation:decomp_uniqueness}
\end{equation}
Hence, if such a $u$ exists with full Schmidt-rank, the completely positive maps $\Psi_0,\Psi_1$ satisfying Equation \eqref{equation:decomposition_properties} are unique (if they exist).
\begin{proof}
We have
\begin{align}
	\vertiii{\Phi}_1 &= \norm{(\Phi\otimes \id{n})(uu^*)}_1\\ 
		&= \norm{(\Psi_0\otimes \id{n})(uu^*) - (\Psi_1\otimes \id{n})(uu^*)}_1 \\
		&\leq \norm{(\Psi_0\otimes \id{n})(uu^*)}_1 + \norm{(\Psi_1\otimes \id{n})(uu^*)}_1 \label{equation:hahn_decomp} \\
		&\leq \vertiii{\Psi_0}_1 + \vertiii{\Psi_1}_1 \\
		&= \vertiii{\Phi}_1.
\end{align}
Hence, all inequalities are equalities, and therefore $\vertiii{\Psi_0}_1 = \norm{(\Psi_0\otimes \id{n})(uu^*)}_1$ and $\vertiii{\Psi_1}_1 = \norm{(\Psi_1\otimes \id{n})(uu^*)}_1$. 

Next, as $\Psi_0$ and $\Psi_1$ are completely positive, it holds that $(\Psi_0\otimes \id{n})(uu^*) \geq 0$ and $(\Psi_1\otimes \id{n})(uu^*) \geq 0$, and so equality in Equation \eqref{equation:hahn_decomp} implies that
\begin{equation}
	(\Phi\otimes \id{n})(uu^*) = (\Psi_0\otimes \id{n})(uu^*) - (\Psi_1\otimes \id{n})(uu^*)
\end{equation}
is the Hahn decomposition of $(\Phi\otimes \id{n})(uu^*)$.\footnote{The Hahn decomposition of a Hermitian matrix $H \in M_n$ is the unique decomposition of $H$ as a difference $H=P-Q$ with $P,Q \geq 0$ and $PQ = 0$. For $H$ Hermitian and $P,Q \geq 0$, it holds that $H=P-Q$ is the Hahn decomposition of $H$ if and only if $\norm{H}_1 = \norm{P}_1 + \norm{Q}_1$.} Thus, for any other completely positive maps $\Psi_0',\Psi_1' : M_n \rightarrow M_m$ satisfying the hypotheses,
\begin{equation}
	(\Phi\otimes \id{n})(uu^*) = (\Psi_0'\otimes \id{n})(uu^*) - (\Psi_1'\otimes \id{n})(uu^*)
\end{equation}
is also the Hahn decomposition of $(\Phi\otimes \id{n})(uu^*)$. Equation \eqref{equation:decomp_uniqueness} therefore follows by the uniqueness of the Hahn decomposition.

Finally, if $u \in \complex^n \otimes \complex^n$ has full Schmidt-rank, then a linear map $\Gamma :M_n \rightarrow M_m$ is uniquely specified by the matrix $(\Gamma \otimes \id{n})(uu^*)$, and so Equation \eqref{equation:decomp_uniqueness} implies the uniqueness of the pair $\Psi_0$ and $\Psi_1$ (assuming such a pair exists).
\end{proof}
\end{lemma}

\begin{proof}[Proof of Theorem \ref{theorem:werner_holevo_unique}]
In both cases the ``if'' part is a matter of verifying Equation \eqref{equation:equal_probability}, where the case $m \geq 2n$ requires use of the fact that $\Psi_0$ and $\Psi_1$ have orthogonal ranges. 

Thus, assume we have a channel discrimination triple $(\lambda, \Gamma_0, \Gamma_1)$ satisfying Equation \eqref{equation:maximal_gap_channel_difference}. By Corollary \ref{corollary:maximal_norm_gap}, the norm relation implies
\begin{equation}
	\lambda \Gamma_0 - (1-\lambda) \Gamma_1 = \frac{1}{n} \Psi T_n, \label{equation:CP_decomp}
\end{equation}
for $\Psi : M_n \rightarrow M_m$ a Hermiticity preserving complete trace-norm isometry. Remark \ref{remark:herm_trace_norm_pres} gives the following structure for $\Psi$:
\begin{itemize}
	\item If $m < 2n$, either $\Psi$ or $-\Psi$ is a reversible quantum channel.
	\item If $m \geq 2n$, there exists $r \in [0,1]$ and $\Psi_0, \Psi_1 : M_n \rightarrow M_m$ reversible quantum channels with orthogonal ranges for which $\Psi = r \Psi_0 - (1-r) \Psi_1$.
\end{itemize}
In what follows we will work with the form of $\Psi$ in the case $m \geq 2n$, as the case $m < 2n$ can be subsumed by the case $r=0$ or $r=1$ when $m \geq 2n$, even though it is not possible for two reversible channels $\Psi_0, \Psi_1 : M_n \rightarrow M_m$ to have orthogonal ranges when $m < 2n$.

Observe the following facts: 
\begin{itemize}
	\item $\frac{1}{n}\Psi T_n$ is Hermiticity preserving and decomposes as a difference of CP maps as given in Equation \eqref{equation:CP_decomp},
	\item $\bigvertiii{\frac{1}{n}\Psi T_n}_1 = 1 = \vertiii{\lambda \Gamma_0}_1 + \vertiii{(1-\lambda) \Gamma_1}_1$, and 
	\item $\bigvertiii{\frac{1}{n}\Psi T_n}_1 = \bignorm{\frac{1}{n} (\Psi T_n \otimes \id{n})(\tau_n)}_1$, where $\tau_n = \frac{1}{n} \sum_{a,b=1}^n E_{a,b} \otimes E_{a,b}  \in M_n \otimes M_n$ is the canonical maximally entangled state.
\end{itemize}
When taken together these facts imply, by Lemma \ref{lemma:CP_difference}, that Equation \eqref{equation:CP_decomp} is the \emph{unique} decomposition of $\frac{1}{n} \Psi T_n$ into a difference of CP maps with the above properties. In the remainder of the proof, we will exhibit a (seemingly) different decomposition of $\frac{1}{n}\Psi T_n$, verify that it also satisfies the assumptions of Lemma \ref{lemma:CP_difference}, then conclude that the two decompositions are necessarily the same.

Note that $\frac{1}{n} T_n = \lambda_n \Phi_0^{(n)} - (1-\lambda_n)\Phi_1^{(n)}$, and hence
\begin{align}
	\frac{1}{n}\Psi T_n &= (r\Psi_0 - (1-r)\Psi_1)(\lambda_n \Phi_0^{(n)} - (1-\lambda_n) \Phi_1^{(n)}) \\
									&= \big[r \lambda_n \Psi_0 \Phi_0^{(n)} + (1-r)(1-\lambda_n) \Psi_1 \Phi_1^{(n)}\big] - \big[(1-r)\lambda_n \Psi_1\Phi_0^{(n)} + r(1-\lambda_n)\Psi_0 \Phi_1^{(n)}\big].
\end{align}
The maps in the square brackets are completely positive, and satisfy
\begin{align}
	\Bigvertiii{r \lambda_n \Psi_0 \Phi_0^{(n)} + (1-r)(1-\lambda_n) \Psi_1 \Phi_1^{(n)}}_1 + \Bigvertiii{(1-r)\lambda_n \Psi_1\Phi_0^{(n)} + r(1-\lambda_n)\Psi_0 \Phi_1^{(n)}}_1\\ = r \lambda_n + (1-r)(1-\lambda_n) + (1-r)\lambda_n + r(1-\lambda_n) = 1.
\end{align}
Hence, by the uniqueness clause of Lemma \ref{lemma:CP_difference}, Equations \eqref{equation:gamma_0} and \eqref{equation:gamma_1} hold.

When $m < 2n$, either $r = 0$ or $r=1$, in which case either 
\begin{equation}
	(\lambda, \Gamma_0, \Gamma_1) = \big(\lambda_n, \Psi_0 \Phi_n^{(0)}, \Psi_0\Phi_n^{(1)}\big)\textnormal{, or }(\lambda, \Gamma_0, \Gamma_1) = \big(1- \lambda_n, \Psi_1 \Phi_n^{(1)}, \Psi_1\Phi_n^{(0)}\big),
\end{equation} 
as required.
\end{proof}

\section{Discussion} \label{section:discussion}

The canonical example of a linear map $\Phi : M_n \rightarrow M_m$ satisfying $\norm{\Phi}_1 = 1$ and
\begin{equation}
	\norm{\Phi \otimes \id{k}}_1 = k
\end{equation}
is the matrix transpose, and we have proven that, up to an equivalence, the transpose is the \emph{unique} map satisfying the above equation. We have applied this result in the setting of single-shot quantum channel discrimination to prove that a channel discrimination game $(\lambda, \Gamma_0, \Gamma_1)$ (with input dimension $n$) satisfies the norm relation
\begin{equation}
	1 = \vertiii{\lambda \Gamma_0 - (1-\lambda)\Gamma_1}_1 = n \norm{\lambda \Gamma_0 - (1-\lambda)\Gamma_1}_1
\end{equation}
if and only if it is in some sense equivalent to the game $(\lambda_n, \Phi_n^{(0)}, \Phi_n^{(1)})$, where $\Phi_n^{(0)},\Phi_n^{(1)}$ are the Werner-Holevo channels, and $\lambda_n = \frac{n+1}{2n}$.

The uniqueness result for the Werner-Holevo channel discrimination game is almost, but not quite, a characterization of channel discrimination games with maximal gap between the optimal performance of entangled and unentangled strategies. Characterizing such games requires an understanding of the maximal gap between $\norm{\Phi}_{1,H}$ and $\vertiii{\Phi}_1$ for Hermiticity preserving linear maps $\Phi :M_n \rightarrow M_m$. For example, is it true that $\vertiii{\Phi}_1 \leq n \norm{\Phi}_{1,H}$? More generally,  is it true that
\begin{equation}
	\norm{\Phi \otimes \id{k}}_{1,H} \leq k \norm{\Phi}_{1,H}?
\end{equation} 
It is not so clear if the proof of the inequality in Theorem \ref{theorem:maximum_norm_gap} can be adapted to this situation. It seems reasonable to conjecture that some inequality of the above form holds, and that the transpose will uniquely saturate the inequality.

Another natural question is whether the characterization of linear maps $\Phi : M_n \rightarrow M_m$ for which $\norm{\Phi}_1 = 1$ and $\norm{\Phi \otimes \id{k}}_1 = k$ holds approximately. E.g. if $\norm{\Phi}_1 = 1$ and $\vertiii{\Phi}_1 \geq n - \epsilon$, then is $\Phi$ necessarily close in some sense to the transpose (followed by a complete trace-norm isometry)?

\subsection*{Acknowledgements}
We thank John Watrous, Vern Paulsen, and Chi-Kwong Li for helpful discussions and comments. This work was supported by NSERC, the Ontario Graduate Scholarship, and the Queen Elizabeth II Graduate Scholarship in Science and Technology. 

\appendix
\section{Duality of the operator norm and trace-norm} \label{appendix:duality}
Let $\norm{\cdot}$ denote the operator norm for matrices, as well as the induced operator norm for linear maps of matrices, i.e. for linear $\Phi : M_n \rightarrow M_m$, 
\begin{equation}
	\norm{\Phi} = \max\{\norm{\Phi(X)} : X \in M_n, \norm{X} = 1\}.
\end{equation}
For $A,B \in M_n$, we denote the Hilbert-Schmidt inner product as $\ip{A}{B} = \Tr(A^*B)$, and for a linear map $\Phi : M_n \rightarrow M_m$, we use $\Phi^*$ to denote the adjoint of $\Phi$ with respect to this inner-product. That is, $\Phi^* : M_m \rightarrow M_n$ is the unique linear map satisfying
\begin{equation}
	\ip{A}{\Phi(B)} = \ip{\Phi^*(A)}{B}
\end{equation}
for all $A \in M_m$ and $B \in M_n$.

For our purposes, the ``duality'' of the trace-norm and operator norm may be summarized by the following: For any matrix $A \in M_n$, it holds that
\begin{equation}
	\norm{A}_1 = \max\{ |\ip{X}{A}| : X \in M_n, \norm{X} \leq 1\},
\end{equation}
and
\begin{equation}
	\norm{A} = \max\{ |\ip{X}{A}| : X \in M_n, \norm{X}_1 \leq 1\}.
\end{equation}
A direct implication of these expressions is that, for any linear $\Phi : M_n \rightarrow M_m$, it holds that
\begin{equation}
	\norm{\Phi}_1 = \norm{\Phi^*}\textnormal{ and } \norm{\Phi} = \norm{\Phi^*}_1.
\end{equation}

The above relations enable interconversion of facts about the trace-norm and facts about the operator norm. For example, the statement 
\begin{equation}
	\norm{\Phi \otimes \id{k}}_1 \leq k \norm{\Phi}_1 \textnormal{ for all linear maps }\Phi : M_n \rightarrow M_m \label{equation:statement_trace}
\end{equation}
is equivalent to the statement 
\begin{equation}
	\norm{\Phi \otimes \id{k}} \leq k \norm{\Phi} \textnormal{ for all linear }\Phi : M_n \rightarrow M_m. \label{equation:statement_operator}
\end{equation}
This is why \cite[Exercise 3.10]{paulsen_completely_2003}, which directly generalizes the statement in Equation \eqref{equation:statement_operator} to arbitrary linear maps between unital C$^*$-algebras, is referenced in the introduction as a generalization of the statement in Equation \eqref{equation:statement_trace}. Similarly, our main result characterizing the linear maps $\Phi : M_n \rightarrow M_m$ for which $\norm{\Phi \otimes \id{k}}_1 =k \norm{\Phi}_1$, may also be translated into a characterization of such maps for which $\norm{\Phi \otimes \id{k}} = k \norm{\Phi}$.

\bibliographystyle{plainnat}
\bibliography{quantum_version}

\end{document}